\documentclass[letterpaper,10pt,conference]{ieeeconf}

\IEEEoverridecommandlockouts
\usepackage{amsmath,amssymb,mathtools}
\usepackage{bm}
\usepackage{graphicx}
\usepackage{epstopdf}
\usepackage{psfrag,color}

\usepackage{tikz}

\usepackage{xcolor}
\usepackage{cite}
\usepackage{hyperref}

\hypersetup{
    colorlinks=true,
    linkcolor=blue,
    citecolor=blue,
    filecolor=magenta,
    urlcolor=cyan
}

\newtheorem{assumption}{Assumption}
\newtheorem{theorem}{Theorem}
\newtheorem{lemma}{Lemma}
\newtheorem{corollary}{Corollary}
\newtheorem{remark}{Remark}
\newtheorem{definition}{Definition}
\newtheorem{proposition}{Proposition}

\title{\LARGE \bf Online Approximate Dynamic Programming within\\ Linear Model Validity Regions Exploiting Impulsive Braking Inputs}

\author{Adebayo Olayinka Oke and Nilay Kant
\thanks{The authors are with the Department of Mechanical and Aerospace Engineering, Missouri University of Science and Technology, Rolla, MO, 65409, USA.}
\thanks{The corresponding author is Nilay Kant: 
	{\tt\small {nilaykant@mst.edu}}}}%
\begin{document}

\maketitle

\begin{abstract}
Persistent excitation required for online approximate dynamic programming (ADP) can drive a nonlinear system beyond the region where its local linear model remains valid. For a class of second-order nonlinear systems , this paper develops an impulse-supervised exploration and learning framework that retains the state within a prescribed linear model validity region (LMVR) while permitting any desired persistent excitation for critic learning. The online learning dynamics are formulated as a hybrid system, and positive invariance of the prescribed exploration set within LMVR is established by resetting the state into its interior through impulsive braking, followed by recovery before learning resumes. Simulations on a nonlinear system demonstrate exploration within the LMVR despite high-amplitude probing excitation and convergence of the learned weights to that given by algebraic Riccati equation.
\end{abstract}

\section{Introduction}
Reinforcement learning (RL) and approximate dynamic programming (ADP) provide frameworks for solving optimal control problems in dynamical systems \cite{lewis2009reinforcement,wang2009adaptive}. In continuous time, ADP learns an optimal control policy by approximating the solution of the Hamilton-Jacobi-Bellman (HJB) equation, typically through an adaptive critic representation of the optimal value function \cite{vamvoudakis2010online,lewis2009reinforcement,wang2009adaptive}. For nonlinear systems, the optimal value function is generally unknown and approximated using neural networks, whereas for LTI systems it is quadratic and characterized by the algebraic Riccati equation, reducing learning to estimation of a set of critic weights \cite{lewis2009reinforcement}.

Approximate linear models and linearization-based controllers remain widely used for nonlinear systems in practice. In this setting, online ADP can be used to learn a locally optimal linear controller, provided that the state remains within a neighborhood where the nonlinear dynamics are sufficiently well represented by their linearization. A key requirement for critic learning, however, is sufficient excitation of the regressor. Persistent excitation (PE) may induce large state excursions \cite{vamvoudakis2010online}, potentially driving the trajectory outside the local model validity region and invalidating the model used for learning. Several approaches have sought to reduce reliance on PE along the actual system trajectory. Concurrent learning combines current and recorded data \cite{chowdhary2010concurrent}, experience replay reuses previously collected data for critic adaptation \cite{kokolakis2022online}, and model-based simulation-of-experience methods evaluate learning information away from the measured trajectory \cite{kamalapurkar2016model,cohen2023safe}. These approaches reduce the excitation burden on the physical system, but introduce additional requirements on stored data, model information, or computation and more importantly, do not address trajectory confinement to a local linear model validity region (LMVR).\

This paper addresses this limitation using impulsive braking inputs (IBI) as a supervisory mechanism during online ADP. When the trajectory reaches the boundary of the linear model exploration region with outward motion, an impulsive braking input resets the state toward its interior. If the reset does not return the state sufficiently far inside the region, learning is temporarily suspended and recovery dynamics are used until the state reaches an inner region after which learning resumes.  The objective differs from safe RL, where forward invariance of a prescribed safe set is typically enforced through constrained RL \cite{ye2024adp,Li2023,yang2019safety,liu2025safe}, barrier-function-based methods \cite{ames2019control,cohen2020approximate}, or safe exploration strategies \cite{GARG2025106458}. Here, confinement is imposed to preserve validity of the local model used for learning rather than to enforce safety. Impulsive control is particularly suitable for this supervisory role and has been extensively studied in mechanical systems, including experimental implementations \cite{jafari2015enlarging,kant2021stabilization,kant2024EHGO,KANT2020104813,kant2025optimal,kant5682696tracking}.

The remainder of this paper is organized as follows. Section II presents the problem formulation and continuous-time ADP framework. Section III characterizes the effect of impulsive inputs. Section IV develops the confined exploration framework and establishes the main theoretical results. Section V presents the simulation results, and Section VI concludes the paper.



\noindent\textbf{Notation:} For a vector $x$, $\|x\|$ denotes the Euclidean norm. For a matrix $A$, $A^\top$ denotes its transpose, $\lambda_{\min}(A)$ its minimum eigenvalue, and $A\succ0$ ($A\succeq0$) positive (semi)definiteness. For a set $\mathcal{S}$, $\mathring{\mathcal{S}}$ and $\partial\mathcal{S}$ denote its interior and boundary, and $\operatorname{dist}(\mathcal{A},\mathcal{B})$ the Euclidean distance between sets. The symbols $x^-$ and $x^+$ denote the pre- and post-braking states.

\section{Problem Formulation}
\subsection{System Description}\label{sec2a}

\noindent Consider the class of second-order nonlinear systems
\begin{subequations}\label{eq:nonl_system_2D}
\begin{align}
\dot{x}_1 &= x_2
\label{eq:nonl_system_2D_a}\\
\dot{x}_2 &= f(x_1,x_2)+bu
\label{eq:nonl_system_2D_b}
\end{align}
\end{subequations}

\noindent where $x_1,x_2\in\mathbb{R}$ are the state variables, $u\in\mathbb{R}$ is the control input, and $b\neq0$ is known. Defining $x\triangleq[x_1\;x_2]^\top\in\mathbb{R}^2$, \eqref{eq:nonl_system_2D} can be written as
\begin{equation}\label{eq:nonlinear_system}
\dot{x}=F(x)+Bu
\end{equation}

\noindent where
\begin{equation}
F(x)\triangleq\begin{bmatrix}x_2\\f(x_1,x_2)\end{bmatrix}, \qquad B\triangleq\begin{bmatrix}0\\b\end{bmatrix}
\end{equation}

\noindent The state $x$ and signal $\dot{x}_2$ are assumed available for measurement. The drift vector field $F(x)$ is unknown or only partially known and is assumed to be continuously differentiable in a neighborhood of the origin and locally Lipschitz. Furthermore, $F(0)=0$, so that the origin is an equilibrium of the unforced system:
\begin{equation}\label{eq:unforced_system}
\dot{x}=F(x)
\end{equation}

\noindent In a neighborhood of the origin,
\begin{equation}\label{eq:local_decomposition}
F(x)=Ax+\phi(x)
\end{equation}

\noindent where
\begin{equation}
A\triangleq\left.\frac{\partial F}{\partial x}\right|_{x=0}, \qquad \phi(x)\triangleq F(x)-Ax
\end{equation}

\noindent By differentiability of $F$ at the origin, we can write \cite{khalil2002nonlinear}
\begin{equation}\label{eq:remainder_property}
\phi(0)=0, \qquad \lim_{\|x\|\rightarrow0}\frac{\|\phi(x)\|}{\|x\|}=0
\end{equation}

\noindent Substituting \eqref{eq:local_decomposition} into \eqref{eq:nonlinear_system} gives
\begin{equation}\label{eq:nonlinear_local_model}
\dot{x}=Ax+Bu+\phi(x)
\end{equation}

\noindent with corresponding local linear model
\begin{equation}\label{eq:linear_model}
\dot{x}=Ax+Bu
\end{equation}

\noindent Since $\phi(x)=o(\|x\|)$ as $x\rightarrow0$, for any prescribed modeling tolerance $\bar{\gamma}>0$, there exists a sufficiently small $r>0$ such that
\begin{equation}\label{eq:remainder_bound}
\|\phi(x)\|\leq\bar{\gamma}\|x\|, \qquad \forall x\in\mathcal{B}_r
\end{equation}

\noindent where
\begin{equation}\label{eq:LMVR}
\mathcal{B}_r\triangleq\left\{x\in\mathbb{R}^2:\|x\|\leq r\right\}
\end{equation}

\begin{definition}[Linear Model Validity Region]\label{def:LMVR}
For a prescribed sufficiently small modeling tolerance $\bar{\gamma}>0$, a set $\mathcal{B}_r$ satisfying \eqref{eq:remainder_bound} is referred to as a \emph{linear model validity region} (LMVR). Within $\mathcal{B}_r$, the nonlinear dynamics are approximated by the linear model \eqref{eq:linear_model} with relative modeling error bounded by $\bar{\gamma}$.
\end{definition}

\noindent Within the LMVR, the modeling error between the nonlinear system \eqref{eq:nonlinear_local_model} and the linear model \eqref{eq:linear_model} satisfies
\begin{equation}\label{eq:model_error_bound}
\|\dot{x}-(Ax+Bu)\|=\|\phi(x)\|\leq\bar{\gamma}\|x\|
\end{equation}

\noindent By selecting $r$ sufficiently small, the relative modeling error within the LMVR can be made arbitrarily small, and the nonlinear dynamics are therefore approximated arbitrarily closely by the linear model \eqref{eq:linear_model}.

\begin{proposition}\label{prop:local_unforced_stability}
Suppose that $A$ is Hurwitz and let $P_L=P_L^\top\succ0$ satisfy
\begin{equation}\label{eq:lyapunov_equation}
A^\top P_L+P_LA=-Q_L, \qquad Q_L=Q_L^\top\succ0
\end{equation}

\noindent Consider the Lyapunov function
\begin{equation}\label{eq:local_lyapunov}
V_L(x)=x^\top P_Lx
\end{equation}

\noindent If
\begin{equation}\label{eq:LMVR_stability_condition}
\bar{\gamma}<\frac{\lambda_{\min}(Q_L)}{2\|P_L\|}
\end{equation}

\noindent and $r>0$ is selected such that \eqref{eq:remainder_bound} holds on $\mathcal{B}_r$, then
\begin{equation}
\dot{V}_L(x)<0, \qquad \forall x\in\mathcal{B}_r\setminus\{0\}
\end{equation}

\noindent Hence, the origin of the unforced nonlinear system \eqref{eq:unforced_system} is locally asymptotically stable.
\end{proposition}

\begin{proof}
See \cite[Theorem 4.7 proof]{khalil2002nonlinear}.
\end{proof}

\begin{remark}
Proposition \ref{prop:local_unforced_stability} establishes local asymptotic stability of the unforced nonlinear dynamics within the LMVR. It does not imply positive invariance of the LMVR under the controlled learning dynamics.
\end{remark}

\subsection{Continuous-Time Approximate Dynamic Programming}\label{sec2b}

\noindent For the system in \eqref{eq:nonlinear_system}, consider the infinite-horizon quadratic cost
\begin{equation}\label{eq:cost}
J(x,u)=\int_0^\infty\left[x(\tau)^\top Qx(\tau)+u(\tau)^\top Ru(\tau)\right]d\tau
\end{equation}

\noindent where $Q\in\mathbb{R}^{2\times2}$ is symmetric positive definite and $R\in\mathbb{R}$ is positive. The optimal value function is
\begin{equation}\label{eq:optimal_cost}
J^*(x)=\inf_{u(\cdot)\in\mathcal{U}}J(x,u)
\end{equation}

\noindent where $\mathcal{U}$ denotes the set of admissible control inputs for which the state trajectory exists, remains bounded, and yields finite cost. The Hamiltonian associated with \eqref{eq:nonlinear_system} and \eqref{eq:cost} is
\begin{equation}\label{eq:hamiltonian}
H(x,u,\nabla J)=x^\top Qx+u^\top Ru+\nabla J(x)^\top[F(x)+Bu]
\end{equation}

\noindent Since $F(x)+Bu=\dot{x}$ along the system trajectories,
\begin{equation}\label{eq:hamiltonian_measured}
H(x,u,\nabla J)=x^\top Qx+u^\top Ru+\nabla J(x)^\top\dot{x}
\end{equation}

\noindent This form is useful because $F(x)$ is unknown, while the state and acceleration are available for measurement. 

The optimal value function satisfies the HJB equation:
\begin{equation}\label{eq:hjb}
\min_{u\in\mathcal{U}}H\bigl(x,u,\nabla J^*(x)\bigr)=0
\end{equation}

\noindent Minimization of \eqref{eq:hamiltonian} with respect to $u$ gives
\begin{equation}\label{eq:optimal_policy}
u^*(x)=-\frac{1}{2}R^{-1}B^\top\nabla J^*(x)
\end{equation}

\noindent Since the exact solution of \eqref{eq:hjb} is generally unavailable for nonlinear systems, ADP approximates the optimal value function online \cite{lewis2009reinforcement,vamvoudakis2010online,wang2009adaptive}. On a compact set,
\begin{equation}\label{eq:value_error}
J^*(x)=W^\top\sigma(x)+\varepsilon(x)
\end{equation}

\noindent where $\sigma(x)\in\mathbb{R}^N$ is a vector of continuously differentiable basis functions, $W\in\mathbb{R}^N$ is the ideal critic weight vector, and $\varepsilon(x)$ is the residual approximation error \cite{vamvoudakis2010online}. The online critic approximation is
\begin{equation}\label{eq:value_approx}
\hat{J}(x)=\hat{W}^\top\sigma(x)
\end{equation}

\noindent where $\hat{W}\in\mathbb{R}^N$ is the estimated critic weight vector, with
\begin{equation}\label{eq:grad_value_estimate}
\nabla\hat{J}(x)=\left(\frac{\partial\sigma(x)}{\partial x}\right)^\top\hat{W}
\end{equation}

\noindent Replacing $\nabla J^*(x)$ in \eqref{eq:optimal_policy} with \eqref{eq:grad_value_estimate} gives the approximate policy
\begin{equation}\label{eq:approx_policy}
\hat{u}(x)=-\frac{1}{2}R^{-1}B^\top\nabla\hat{J}(x)
\end{equation}

\noindent Substitution into the Hamiltonian yields the Bellman residual
\begin{equation}\label{eq:bellman_residual_compact}
e(x,\hat{W})=x^\top Qx+\hat{u}^\top(x)R\hat{u}(x)+\hat{W}^\top\alpha(x,\dot{x})
\end{equation}

\noindent where the critic regressor is
\begin{equation*}
\alpha(x,\dot{x})\triangleq\frac{\partial\sigma(x)}{\partial x}\dot{x}
\end{equation*}

The critic weights are updated using the normalized gradient law \cite{vamvoudakis2010online}
\begin{equation}\label{eq:weight_update}
\dot{\hat{W}}=-\eta\frac{\alpha(x,\dot{x})}{\left[1+\alpha^\top(x,\dot{x})\alpha(x,\dot{x})\right]^2}e(x,\hat{W})
\end{equation}

\noindent where $\eta>0$ is the adaptation gain. Persistent excitation (PE) of the critic regressor is required for convergence of the critic weights \cite{vamvoudakis2010online}.

\subsection{Problem Statement}\label{sec2c}

\begin{definition}[Linear Model Exploration Set (LMES)]\label{def:lmes}
Let
\begin{equation}\label{eq:circular_function}
V(x)=x^\top x=x_1^2+x_2^2
\end{equation}

\noindent and define the \emph{linear model exploration set} (LMES) as
\begin{equation}\label{eq:outer_circle}
\mathcal{S}_1\triangleq\left\{x\in\mathbb{R}^2:V(x)\leq r_1^2\right\}
\end{equation}

\noindent where $r_1>0$ is selected such that $\mathcal{S}_1\subseteq\mathcal{B}_r$.
\end{definition}

Since $\mathcal{S}_1\subseteq\mathcal{B}_r$, the local linear model \eqref{eq:linear_model} approximates the nonlinear dynamics within the prescribed modeling tolerance throughout the LMES.

\noindent For the local linear model \eqref{eq:linear_model}, the optimal value function associated with \eqref{eq:optimal_cost} is quadratic and can be written as
\begin{align}\label{eq:linear_weights_ideal}
J^*(x) &= W^\top\sigma(x) \\
W=\begin{bmatrix}w_1&w_2&w_3\end{bmatrix}^\top,& \quad \sigma(x)=\begin{bmatrix}x_1^2&x_1x_2&x_2^2\end{bmatrix}^\top \nonumber
\end{align}

\noindent The ideal weight vector $W$ is determined by the corresponding algebraic Riccati equation \cite{lewis2009reinforcement}.

Persistent excitation (PE) is required for convergence of $\hat{W}$ to $W$ under \eqref{eq:weight_update}. However, sufficiently rich excitation may induce large state excursions \cite{vamvoudakis2010online}, potentially driving the trajectory outside the LMVR where the local linear model is no longer accurate. The objective of this paper is to learn the optimal policy associated with \eqref{eq:linear_model} while confining the nonlinear system trajectory to LMES $\mathcal{S}_1$. To achieve this, impulsive braking inputs are introduced as a supervisory mechanism to render the LMES positively invariant.

\section{Effect of Impulsive Braking on the Dynamics}\label{Sec2c}

Suppose the control input $u(t)$ in \eqref{eq:nonlinear_system} contains impulsive inputs applied at isolated time instants $t=\tau_k$, $k=1,2,\dots$. At $t=\tau_k$,
\begin{equation}\label{eq:impulsive_input}
u(t)=I_k\delta(t-\tau_k)
\end{equation}
where $I_k\in\mathbb{R}$ denotes the impulse magnitude and $\delta(\cdot)$ is the Dirac delta function.

To characterize the state jump, integrate \eqref{eq:nonlinear_system} over $[\tau_k^-,\tau_k^+]$
\begin{equation}\label{eq:x_integral}
\int_{\tau_k^-}^{\tau_k^+}\dot{x}(t)\,dt=\int_{\tau_k^-}^{\tau_k^+}F(x)\,dt+\int_{\tau_k^-}^{\tau_k^+}Bu(t)\,dt
\end{equation}

\noindent Let $x^-\triangleq x(\tau_k^-)$ and $x^+\triangleq x(\tau_k^+)$. Since $F(x)$ is locally bounded, its contribution vanishes over the infinitesimal interval $\tau_k^+ \!-\! \tau_k^-$. Using the sifting property of the Dirac delta function gives
\begin{equation}\label{eq:impulse_effect}
x^+-x^-=BI_k
\end{equation}

\noindent Since $B=[0\;\;b]^\top$,
\begin{equation}\label{eq:component-jump}
x_1^+=x_1^-, \qquad x_2^+=x_2^-+bI_k
\end{equation}

\noindent Thus, the impulsive input\footnote{Impulsive inputs have been extensively employed for the control of mechanical systems including on hardware; see, for example, \cite{jafari2015enlarging,kant2021stabilization,kant2024EHGO,KANT2020104813,kant2025optimal,kant5682696tracking}.} leaves $x_1$ unchanged and induces an instantaneous jump in $x_2$. Choosing
\begin{equation}
I_k = -(x_2^-/b)
\end{equation}

\noindent gives $x_2^+=0$. The resulting \emph{impulsive braking map} (IBM) is
\begin{equation}\label{eq:braking_jump}
x^+=\begin{bmatrix}x_1^-&0\end{bmatrix}^\top
\end{equation}

\begin{lemma}\label{lemma1}
Consider the IBM in \eqref{eq:braking_jump}. For any $x^-\in\mathbb{R}^2$, the function $V(x)$ in \eqref{eq:circular_function} satisfies
\begin{equation}
V(x^+)\leq V(x^-)
\end{equation}
with equality if and only if $x_2^-=0$.
\end{lemma}

\begin{proof}
Under the IBM, $x_1^+=x_1^-$ and $x_2^+=0$, and therefore
\[
V(x^+)=(x_1^-)^2\leq(x_1^-)^2+(x_2^-)^2=V(x^-)
\]
\end{proof}

\begin{corollary}\label{corr:1}
If $x^-\in\partial\mathcal{S}_1$, then the IBM \eqref{eq:braking_jump} ensures $x^+\in\mathcal{S}_1$.
\end{corollary}

\begin{proof}
Since $V(x^+)\leq V(x^-)=r_1^2$, it follows that $x^+\in\mathcal{S}_1$.
\end{proof}

The IBM is therefore non-expansive with respect to $V(x)$ and strictly decreases $V(x)$ whenever $x_2^-\neq0$. It is used next as a supervisory mechanism to confine the ADP exploration to LMES.

\section{ADP in LMES using Supervisory Impulsive Braking Inputs}\label{sec4}
\subsection{Geometric Overview}\label{sec4a}

The hybrid control architecture, which will be formally presented later is conceptually illustrated here. Fig. \ref{Fig1} depicts the phase portrait of \eqref{eq:nonlinear_system}. The outer set $\mathcal{S}_1$ represents the LMES in \eqref{eq:outer_circle}. Define
\begin{equation}\label{eq:beta-def}
r_2\triangleq r_1\sqrt{\beta}, \qquad \beta\in(0,1)
\end{equation}

\noindent and the inner set
\begin{equation}\label{eq:inner_circle}
\mathcal{S}_2\triangleq\left\{x\in\mathbb{R}^2:V(x)\leq r_2^2\right\}
\end{equation}

\noindent Hence, $\mathcal{S}_2\subset\mathcal{S}_1$. To characterize the effect of impulsive resets on $\partial\mathcal{S}_1$, define
\begin{equation}\label{eq:C1}
\mathcal{C}_1\triangleq\left\{x\in\partial\mathcal{S}_1:x^+\in\mathcal{S}_2\right\}
\end{equation}
and
\begin{equation}\label{eq:C2}
\mathcal{C}_2\triangleq\partial\mathcal{S}_1\setminus\mathcal{C}_1
\end{equation}

\noindent where $x^+$ denotes the post-jump state under impulsive braking. Thus, $\mathcal{C}_1$ contains boundary points whose reset lies in $\mathcal{S}_2$, whereas resets from $\mathcal{C}_2$ remain in $\mathcal{S}_1\setminus\mathcal{S}_2$.

\begin{figure}[t!]
\centering
\psfrag{A}[c][c][1][0]{\small $x_1$}
\psfrag{B}[c][c][1][0]{\small $\mathcal{S}_2$}
\psfrag{C}[c][c][1][0]{\small $x_0$}
\psfrag{D}[c][c][1][0]{\small $B$}
\psfrag{F}[c][c][1][0]{\small $C$}
\psfrag{Q}[c][c][1][0]{\small $D$}
\psfrag{G}[c][c][1][0]{\small $E$}
\psfrag{H}[c][c][1][0]{\small $x_2$}
\psfrag{I}[c][c][1][0]{\small $\mathcal{S}_1$}
\psfrag{J}[c][c][1][0]{\small $A$}
\includegraphics[width=0.8\linewidth]{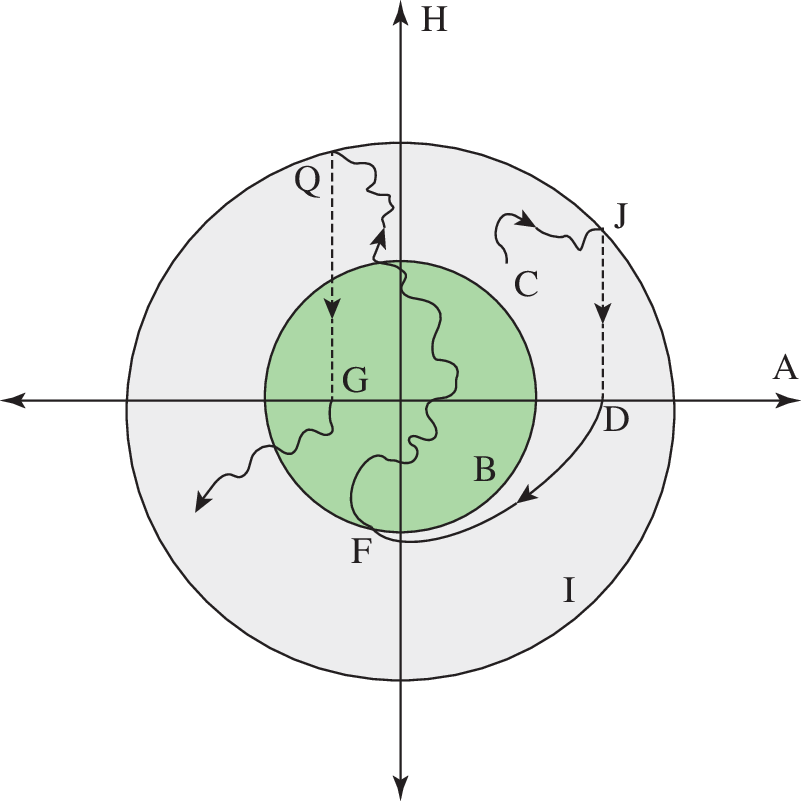}
\caption{Conceptual illustration of the proposed approach.}
\label{Fig1}
\end{figure}

Consider an initial condition $x(t_0)\in\mathring{\mathcal{S}}_1$, as shown in Fig. \ref{Fig1}. During active learning, the control input consists of the ADP policy $\hat{u}(x)$ augmented with a persistently exciting signal. Suppose the trajectory reaches $\partial\mathcal{S}_1$ at point $A\in\mathcal{C}_2$. An impulsive braking input resets the velocity to $x_2^+=0$ according to \eqref{eq:braking_jump}, producing the post-jump state $B\in\mathcal{S}_1\setminus\mathcal{S}_2$. The control and weight adaptation are then suspended while the system evolves under the unforced dynamics $\dot{x}=F(x)$ until it reaches $\mathcal{S}_2$ at point $C$, where active learning resumes.

Suppose the trajectory subsequently reaches $\partial\mathcal{S}_1$ at point $D$. If $D\in\mathcal{C}_1$, the impulsive reset maps the state directly into $\mathcal{S}_2$, and the control and weight adaptation remain active. If $D\in\mathcal{C}_2$, the post-jump state lies in $\mathcal{S}_1\setminus\mathcal{S}_2$, and the recovery mode is activated until the trajectory returns to $\mathcal{S}_2$. Repetition of these continuous flows and impulsive resets confines the exploratory trajectory to $\mathcal{S}_1$, resulting in hybrid dynamics with continuous evolution, state-triggered resets, and mode-dependent learning.

\subsection{Hybrid Automaton}\label{subsec:hybrid_automaton}

The closed-loop system is modeled as the hybrid automaton $\mathcal{H}$ \cite{lygeros2003dynamical}:
\begin{equation}\label{eq:hybrid_automaton}
\mathcal{H}=(Q,X,\mathcal{F},\mathrm{Init},\mathrm{Dom},E,\mathcal{G},\mathcal{R})
\end{equation}

\noindent whose directed graph is shown in Fig. \ref{Fig2}.

\begin{figure}[t!]
\vspace{0.1in}
 \centering
\psfrag{A}[c][c][1][0]{\small $q_1$}
\psfrag{B}[c][c][1][0]{\small $\dot x = F(x) + B u_c$}
\psfrag{C}[c][c][1][0]{\small $x \in \mathcal{S}_2$}
\psfrag{D}[c][c][1][0]{\small $q_2$}
\psfrag{E}[c][c][1][0]{\small $\dot x = F(x) + B u_c$}
\psfrag{F}[c][c][1][0]{\small  $x \in \mathring{\mathcal{S}}_1 \setminus \mathcal{S}_2$}
\psfrag{G}[c][c][1][0]{\small $q_3$}
\psfrag{H}[c][c][1][0]{\small $\dot x = F(x) + B u_c$}
\psfrag{I}[c][c][1][0]{\small  $x \in \partial \mathcal{S}_1$}
\psfrag{J}[c][c][1][0]{\small $q_4$}
\psfrag{K}[c][c][1][0]{\small $\dot x = F(x)$}
\psfrag{L}[c][c][1][0]{\small  $x \in \mathring{\mathcal{S}}_1 \setminus \mathcal{S}_2$}
\psfrag{V}[c][c][1][0]{\small $\mathcal{G}(q_1,q_2)$}
\psfrag{W}[c][c][1][0]{\small $\mathcal{G}(q_2,q_1)$}
\psfrag{M}[c][c][1][0]{\small $\mathcal{G}(q_3,q_2)$}
\psfrag{N}[c][c][1][0]{\small $\mathcal{G}(q_2,q_3)$}
\psfrag{P}[c][c][1][0]{\small $\mathcal{G}(q_3,q_1)$}
\psfrag{T}[c][c][1][0]{\small $x(0)$}
\psfrag{Q}[c][c][1][0]{\small $\mathcal{G}(q_4,q_1)$}
\psfrag{O}[c][c][1][0]{\small $\mathcal{G}(q_3,q_4)$}
\psfrag{R}[c][c][1][0]{\small $\mathcal{R}(q_1,q_2,x)$}
\psfrag{S}[c][c][1][0]{\small $\mathcal{R}(q_2,q_1,x)$}
\psfrag{Z}[c][c][1][0]{\small $\mathcal{R}(q_3,q_1,x)$}
\psfrag{U}[c][c][1][0]{\small $\mathcal{R}(q_2,q_3,x)$}
\psfrag{Y}[c][c][1][0]{\small $\mathcal{R}(q_3,q_4,x)$}
\psfrag{X}[c][c][1][0]{\small $\mathcal{R}(q_4,q_1,x)$}
\psfrag{@}[c][c][1][0]{\small $\mathcal{R}(q_3,q_1,x)$}
\psfrag{#}[c][c][1][0]{\small $x_0 \in \mathcal{S}_2$}
\psfrag{*}[c][c][1][0]{\small $x_0 \in \mathring{\mathcal{S}}_1 \setminus \mathcal{S}_2$}
 \includegraphics[width=0.85\linewidth]{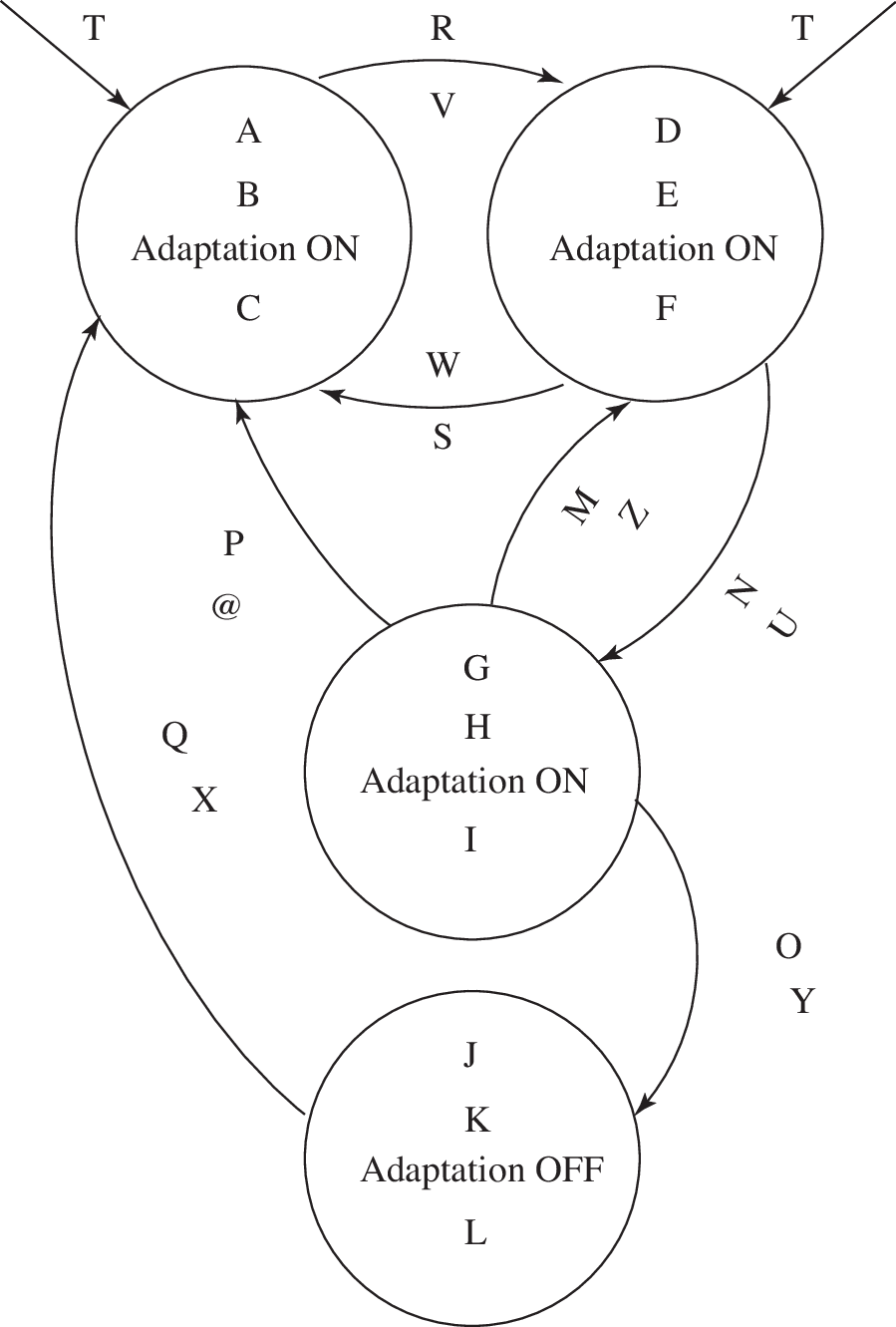}
\caption{The hybrid automaton $\mathcal{H}$.}
 \label{Fig2}
 \end{figure}

\begin{definition}[Learning and Recovery Modes]\label{def:learning_modes}
The discrete mode set is
\[
Q=\{q_1,q_2,q_3,q_4\}
\]
The modes $q_1$, $q_2$, and $q_3$ are active learning modes corresponding to evolution in $\mathcal{S}_2$, $\mathring{\mathcal{S}}_1\setminus\mathcal{S}_2$, and $\partial\mathcal{S}_1$. The mode $q_4$ is the recovery mode, in which the ADP and weight adaptation are paused until the trajectory returns to $\mathcal{S}_2$.
\end{definition}

Let $X=\mathbb{R}^2$ and define the active-learning input as
\begin{equation}\label{eq:input_with_PE}
u(x,t)\triangleq\hat{u}(x)+u_{\mathrm{PE}}(t)
\end{equation}

\noindent where $u_{\mathrm{PE}}(t)$ is a bounded persistently exciting signal and $u(x,t)$ is bounded\footnote{The boundedness of $u$ follows from the boundedness of $\hat W$, established later under Assumption \ref{ass:hybrid_PE}.} on $\mathcal{S}_1$.

The mode-dependent vector fields are
\begin{align*}
\mathcal{F}(q_i,x) &= F(x)+Bu(x,t), \qquad i\in\{1,2,3\}\\
\mathcal{F}(q_4,x) &= F(x)
\end{align*}

The initial set lies in $q_1$ for $x\in\mathcal{S}_2$ and in $q_2$ for $x\in\mathring{\mathcal{S}}_1\setminus\mathcal{S}_2$. The domains are
\begin{alignat*}{2}
\mathrm{Dom}(q_1) &= \mathcal{S}_2, \quad
&\mathrm{Dom}(q_2) &= \mathring{\mathcal{S}}_1\setminus\mathcal{S}_2 \\
\mathrm{Dom}(q_3) &= \partial\mathcal{S}_1, \quad
&\mathrm{Dom}(q_4) &= \mathring{\mathcal{S}}_1\setminus\mathcal{S}_2
\end{alignat*}

\noindent with edge set
\begin{align*}
E &= \{(q_1,q_2),(q_2,q_1),(q_2,q_3),(q_3,q_2),\\
&\quad (q_3,q_1),(q_3,q_4),(q_4,q_1)\}
\end{align*}

\noindent The guard sets are defined using the sign of
\[
\dot V(x)=\nabla V(x)^\top\mathcal{F}(q,x)
\]
as
\begin{align*}
\mathcal{G}(q_1,q_2) &= \left\{x\in\mathbb{R}^2:V(x)\geq r_2^2,\dot{V}(x)>0\right\}\\
\mathcal{G}(q_2,q_1) &= \left\{x\in\mathbb{R}^2:V(x)\leq r_2^2,\dot{V}(x)\leq0\right\}\\
\mathcal{G}(q_2,q_3) &= \left\{x\in\mathbb{R}^2:x\in\partial\mathcal{S}_1\right\}\\
\mathcal{G}(q_3,q_2) &= \left\{x\in\mathbb{R}^2:x\in\partial\mathcal{S}_1,\dot{V}(x)<0\right\}\\
\mathcal{G}(q_3,q_1) &= \left\{x\in\mathcal{C}_1:\dot{V}(x)>0\right\}\\
\mathcal{G}(q_3,q_4) &= \left\{x\in\mathcal{C}_2:\dot{V}(x)>0\right\}\\
\mathcal{G}(q_4,q_1) &= \left\{x\in\mathbb{R}^2:V(x)\leq r_2^2,\dot{V}(x)<0\right\}
\end{align*}

\noindent Thus, the IBI is applied only at $\partial\mathcal{S}_1$ when the continuous dynamics are directed outward. The reset map $\mathcal{R}$ applies the IBM on outward transitions from $q_3$
\begin{align*}
\mathcal{R}(q_3,q_1,x) &= \begin{bmatrix}x_1&0\end{bmatrix}^\top, \quad
\mathcal{R}(q_3,q_4,x) = \begin{bmatrix}x_1&0\end{bmatrix}^\top
\end{align*}

\noindent while all remaining transitions have identity resets
\begin{equation*}
\mathcal{R}(q_i,q_j,x)=x, \quad \forall(q_i,q_j)\in E\setminus\{(q_3,q_1),(q_3,q_4)\}
\end{equation*}

\noindent Following an IBI, the system transitions to $q_1$ if $x^+\in\mathcal{S}_2$ and to $q_4$ otherwise.

\begin{assumption}[Recovery condition]\label{ass:safe_recovery}
Every solution entering the recovery mode $q_4$ reaches $\mathcal{S}_2$ in finite time while remaining in $\mathcal{S}_1$.
\end{assumption}

\begin{remark}
A sufficient condition for Assumption \ref{ass:safe_recovery} is
\[
\dot V(x)<0, \qquad x\in\mathcal{S}_1\setminus\mathcal{S}_2
\]
along the recovery dynamics $\dot{x}=F(x)$.
\end{remark}

\begin{theorem}[Positive invariance of $\mathcal{S}_1$]\label{thm:positive_invariance}
Consider the hybrid automaton $\mathcal{H}$ in \eqref{eq:hybrid_automaton} associated with \eqref{eq:nonlinear_system}, with $x(0)\in\mathring{\mathcal{S}}_1$. Under Assumption \ref{ass:safe_recovery}, the LMES $\mathcal{S}_1$ is positively invariant, i.e.,
\begin{equation*}
x(t)\in\mathcal{S}_1, \qquad \forall t\geq0
\end{equation*}

\noindent Furthermore, there exists a strictly positive lower bound on the time between consecutive braking inputs, thereby excluding Zeno behavior.
\end{theorem}

\begin{proof}
Consider a trajectory reaching $\partial\mathcal{S}_1$. If $\dot V(x)<0$, it re-enters $\mathring{\mathcal{S}}_1$ without a braking reset. If $\dot V(x)=0$, it does not instantaneously leave $\mathcal{S}_1$; a reset is triggered if $\dot V(x)$ becomes positive. For $\dot V(x)>0$, the IBI is applied before outward flow occurs. By Corollary \ref{corr:1},
\begin{equation*}
x(\tau_k^+)\in\mathcal{S}_1
\end{equation*}

\noindent If $x(\tau_k^+)\in\mathcal{S}_2$, active learning resumes immediately. Otherwise, the system enters $q_4$, and Assumption \ref{ass:safe_recovery} guarantees finite-time return to $\mathcal{S}_2$ while remaining in $\mathcal{S}_1$. Hence, neither continuous evolution, impulsive resets, nor recovery can cause the trajectory to leave the LMES $\mathcal{S}_1$, which is therefore positively invariant.

To establish the bound on two consecutive IBI, note that $r_2=r_1\sqrt{\beta}$ from \eqref{eq:beta-def}, and hence
\begin{equation*}
d_\beta\triangleq\operatorname{dist}(\mathcal{S}_2,\partial\mathcal{S}_1)=r_1(1-\sqrt{\beta})>0
\end{equation*}

\noindent Since $\mathcal{S}_1$ is compact and the active-learning input $u(x,t)$ is bounded on $\mathcal{S}_1$, there exists $M>0$ such that
\begin{equation*}
\|\mathcal{F}(q,x)\|\leq M, \qquad q\in\{q_1,q_2,q_3\}
\end{equation*}

\noindent Following each IBI, active learning resumes from $\mathcal{S}_2$, either immediately or after recovery. Therefore, before the IBI can occur, the trajectory must traverse at least the distance $d_\beta$. Thus,
\begin{equation}\label{eq:min_interevent}
\tau_{k+1}-\tau_k\geq\tau_{\min} = \frac{r_1(1-\sqrt{\beta})}{M}>0
\end{equation}

\noindent Hence, consecutive braking inputs are separated by a strictly positive time interval, and Zeno behavior is excluded.
\end{proof}

\subsection{Learning During Hybrid Exploration}
\label{subsec:learning}

Define the learning-mode indicator
\begin{equation}\label{eq:learning_indicator}
\chi(q)\triangleq
\begin{cases}
1, & q\in\{q_1,q_2,q_3\}\\
0, & q=q_4
\end{cases}
\end{equation}

\noindent During active learning, the applied input is $u=\hat{u}+u_{\mathrm{PE}}$, and hence
\begin{equation}
\dot{x}=F(x)+B\hat{u}+Bu_{\mathrm{PE}}
\end{equation}

\noindent Since the Bellman residual in \eqref{eq:bellman_residual_compact} is evaluated with respect to $\hat{u}$, the known probing contribution $Bu_{\mathrm{PE}}$ is removed from the measured state derivative. Define
\begin{equation}\label{eq:supervised_regressor}
\alpha_s\triangleq\frac{\partial\sigma(x)}{\partial x}\left[\dot{x}-Bu_{\mathrm{PE}}\right]
\end{equation}

\noindent The critic update in \eqref{eq:weight_update} is modified as
\begin{equation}\label{eq:hybrid_critic_update}
\dot{\hat{W}}=-\chi(q)\eta\frac{\alpha_s}{\left(1+\alpha_s^\top\alpha_s\right)^2}e(x,\hat{W})
\end{equation}

\noindent so that weight adaptation is suspended during recovery. To account for these interruptions, define the accumulated active-learning time
\begin{equation}\label{eq:active_learning_time}
s(t)\triangleq\int_0^t\chi(q(\tau))d\tau
\end{equation}

\noindent and the normalized regressor
\begin{equation}\label{eq:normalized_regressor}
\bar{\alpha}_s\triangleq\frac{\alpha_s}{1+\alpha_s^\top\alpha_s}
\end{equation}

\begin{assumption}[Persistent excitation in active learning]\label{ass:hybrid_PE}
There exist $T_{\mathrm{PE}}>0$ and $\mu>0$ such that, for every $s^\star$ in the range of $s(t)$, the following holds true:
\begin{equation}\label{eq:hybrid_PE}
G_{\rm PE} = \int_{s^\star}^{s^\star+T_{\mathrm{PE}}}\bar{\alpha}_s(\zeta)\bar{\alpha}_s^\top(\zeta)d\zeta\succeq\mu I
\end{equation}
\end{assumption}

\noindent Assumption \ref{ass:hybrid_PE} is the standard PE condition \cite{vamvoudakis2010online, ioannou1996robust} expressed in active-learning time, with recovery intervals excluded from the learning window. Under this assumption, the critic weights remain bounded and converge to a neighborhood of the ideal weights, with the ultimate weight error depending on the modeling tolerance $\bar{\gamma}$ in \eqref{eq:LMVR}. Consequently, $\hat{u}(x)$, and hence $u(x,t)$ in \eqref{eq:input_with_PE}, remain bounded on $\mathcal{S}_1$.

\section{Simulation Results}
\label{sec:simulation_results}

\begin{figure}[b!]
\centering
\psfrag{A}[c][c][1][0]{\small $x_2$ vs $x_1$}
\includegraphics[width=0.85\linewidth]{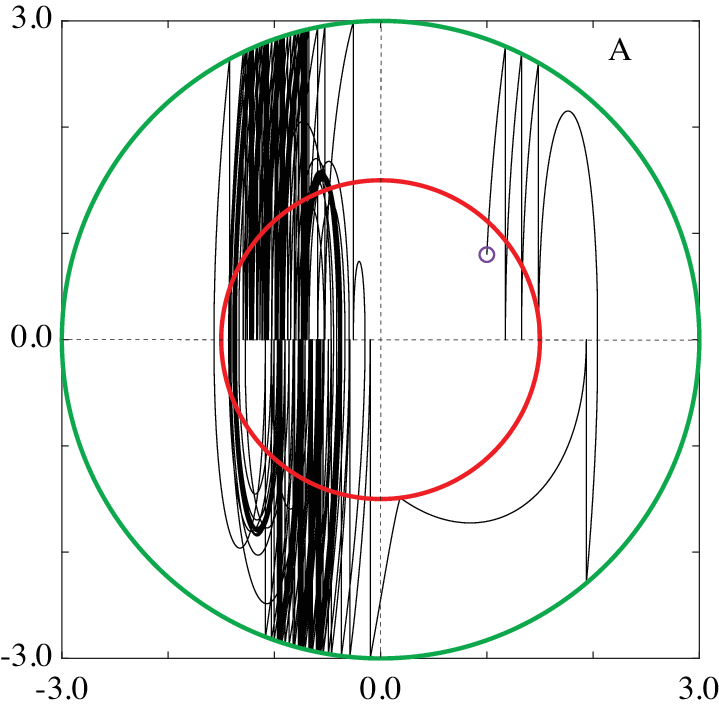}
\caption{The IBI prevents the trajectory from leaving the LMES $\mathcal{S}_1$, while $\mathcal{S}_2$ defines the re-entry boundary for resuming learning.}
\label{fig:confined}
\end{figure}

\begin{figure}[b!]
\centering
\psfrag{A}[c][c][1][0]{\scriptsize $\| W \!-\! W^\star\|_2$}
\psfrag{B}[c][c][1][0]{\small $u_{\rm PE} = 0$}
\includegraphics[width=0.85\linewidth]{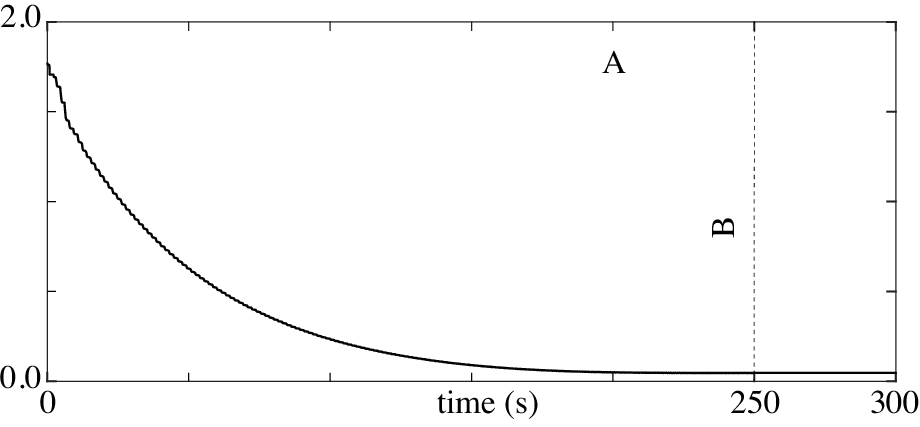}
\caption{Plot showing convergence of the critic weights to $W^\star$ for the proposed approach.}
\label{fig:learning_diagnostics}
\end{figure}

To demonstrate the effectiveness of the proposed framework, a nonlinear mechanical system of the form \eqref{eq:local_decomposition} is considered with
\begin{equation*}
A=
\begin{bmatrix}
0&1\\[-1mm]
-2&-1
\end{bmatrix},\qquad
B=
\begin{bmatrix}
0\\1
\end{bmatrix},\qquad
\phi(x)=
\begin{bmatrix}
0\\
2x_1^3+x_2^3
\end{bmatrix}
\end{equation*}

\noindent $B$ is known, whereas $A$ and $\phi(x)$ are unknown to the learning algorithm. Note that the nonlinear terms in $\phi(x)$ are destabilizing, however, they are negligible near the equilibrium $x = 0$, where the stable linear dynamics dominate.

The cost parameters in \eqref{eq:cost} are selected as $Q=I_2$ and $R=10$. For comparison, the optimal critic weights corresponding to the true local linear model are obtained by solving the algebraic Riccati equation as:
\begin{equation*}
W^\star=[\,1.7111,\;0.4969,\;0.7224\,]^\top
\end{equation*}

The exploration sets in \eqref{eq:outer_circle} and \eqref{eq:inner_circle} are defined by $r_1=0.30$, $\beta=0.25$, and $r_2=0.15$, with initial condition $x(0)=[\,0.10,\;0.08\,]^\top$. The critic adaptation gain is chosen as $\eta=200$. During active learning, the excitation signal in \eqref{eq:input_with_PE} is selected as
\begin{equation*}
u_{\rm PE}(t)=2[\sin(1.3\pi t)+\cos(1.3\pi t)]
\end{equation*}

\noindent The large excitation amplitude is chosen to generate aggressive exploration and is switched off at $t=250$ s.

Fig. \ref{fig:confined} shows that the trajectory remains confined to $\mathcal{S}_1$ (outer green circle) despite the aggressive excitation. At outward encounters with $\partial\mathcal{S}_1$, the IBI resets the velocity to $x_2^+=0$. Learning is suspended during recovery and resumes after the trajectory re-enters $\mathcal{S}_2$ (inner red circle). The maximum observed state norm is $0.30$, confirming the positive invariance of $\mathcal{S}_1$ established in Theorem \ref{thm:positive_invariance}

Fig. \ref{fig:learning_diagnostics} shows convergence of the critic weights to a neighborhood of the local optimal weights, with
\begin{equation*}
\|\hat W-W^*\|_2=0.04635
\end{equation*}

\noindent at $t=250$ s. For $T_{\rm PE}=1$ s, the minimum observed eigenvalue of $G_{\rm PE}$ is obtained as $3.52\times10^{-6}$, providing numerical evidence that the excitation condition in \eqref{eq:hybrid_PE} is satisfied. Furthermore, the minimum observed interval between two consecutive braking is $0.097$ s, consistent with the non-Zeno result in Theorem \ref{thm:positive_invariance}.

\begin{figure}[t!]
\centering
\psfrag{A}[c][c][1][0]{\small $x_2$ vs $x_1$}
\includegraphics[width=0.9\linewidth]{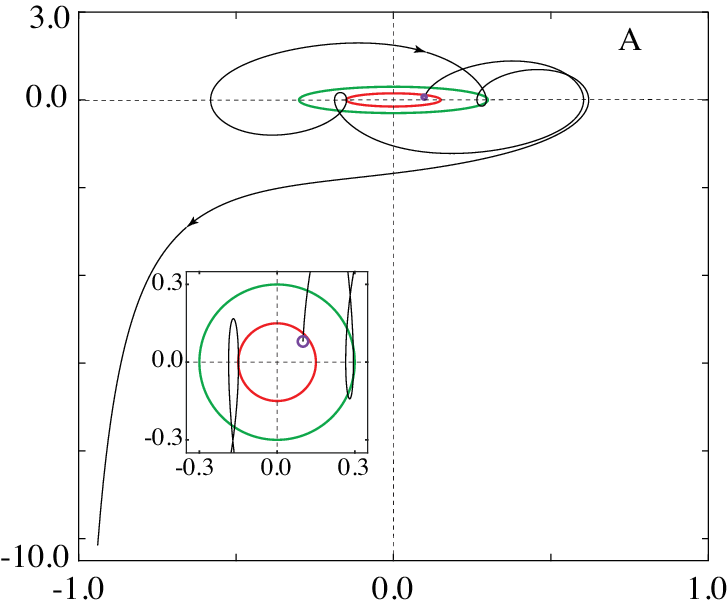}
\caption{Phase portrait for the standard continuous-time ADP.}
\label{fig:unconfined}
\end{figure}

For comparison, Fig. \ref{fig:unconfined} shows the phase portrait under the standard continuous-time ADP in \eqref{eq:approx_policy} and \eqref{eq:weight_update} without supervisory impulsive braking input. It can be observed that the trajectory leaves $\mathcal{S}_1$ and enters a region where the nonlinear terms become significant, eventually resulting in finite-time escape. In contrast, the proposed approach confines exploration to the prescribed LMES in Fig. \ref{fig:confined} while learning the locally optimal linear controller.

\begin{remark}
The simulations use ideal impulsive braking inputs. In practice, such inputs can be approximated using short-duration high-gain control actions or solenoid-based braking mechanisms; see the experimental implementations in \cite{kant2021stabilization,jafari2015enlarging,kant2019estimation,kant2024EHGO}.
\end{remark}
 
\section{Conclusion}
This paper proposes an online approximate dynamic programming (ADP) framework for learning locally optimal control policies for a class of nonlinear systems. By integrating continuous-time ADP with supervisory impulsive braking, the proposed approach enables desired persistent excitation (PE) for learning while confining the state trajectory to the Linear Model Exploration Set (LMES), where a local linear model approximation of the nonlinear system remains valid. The resulting closed-loop system is formulated as a hybrid automaton, and positive invariance of the LMES is established. A positive lower bound on the time between successive impulsive braking inputs is also established, thereby excluding Zeno behavior. Simulation results demonstrate confined exploration, sufficient excitation of the critic regressor, and convergence of the learned critic weights toward the local optimal solution. Future work will investigate relaxation of the recovery condition through alternative recovery mechanisms, relaxing the assumption that $\dot{x}_2$ is directly available for measurement, extensions to higher-dimensional and more general nonlinear systems, and experimental validation. Discrete critic updates following braking inputs and the deriving the associated hybrid PE conditions \cite{saoud2024hybrid} also lie within the scope of future research.

\bibliographystyle{IEEEtran}
\bibliography{references}
\end{document}